%% file: main.tex
\documentclass{llncs}
 
\usepackage{array}
\usepackage[dvips]{graphicx}
\usepackage{epic, eepic}
\usepackage{amssymb,amsmath,stmaryrd}
\usepackage{latexsym}
\usepackage{subfigure}
\usepackage[usenames]{color}
\usepackage{multirow}
\usepackage{gastex}
\usepackage{times}
\usepackage{subfigure}
\let\emptyset\varnothing

\newcounter{compressEnum}
\renewcommand{\thecompressEnum}{$\roman{compressEnum}$}
\newenvironment{compressEnum}
{\setcounter{compressEnum}{0}}
{}
\newcommand{\itCompress}{\stepcounter{compressEnum}{(\thecompressEnum) }}

\newenvironment{myProof}[1][\unskip]{\medskip\par\noindent{\bfseries Proof
    #1.}\ \ 
        \global\def\qed{\origQED\global\def\qed{}}\penalty10000}%
        {\qed\par\medskip\global\def\qed{\origQED\global\def\qed{}}}

\def\endof{%
  \leavevmode
  \parfillskip=0pt%
  \widowpenalty=10000%
  \displaywidowpenalty=10000%
  \finalhyphendemerits=0%
  \unskip\nobreak\null\hfil\penalty50\hskip2em\null\hfill%
}
\def\eodsymbol{\ensuremath\square}
\def\eopsymbol{\ensuremath\blacksquare}

\def\origEOD{\nobreak\leavevmode\endof\eodsymbol\par}
\def\EOD{\origEOD\global\def\EOD{}}
\def\origQED{{\nobreak\leavevmode\endof\eopsymbol\par\medskip}}
\def\qed{\origQED\global\def\qed{}}

\def\abs#1{\ensuremath{\lvert #1\rvert}}

\renewcommand{\P}{\mathbb P}

\newcommand{\nat}{\mathbb N} 
\newcommand{\rat}{{\mathbb Q}}

\newcommand{\real}{{\mathbb R}}

\newcommand{\tuple}[1]{\langle #1 \rangle}

\newcommand{\C}{\mathcal{C}}
\newcommand{\D}{\mathcal{D}}

\newcommand{\Run}{{\sf Run}}

\newcommand{\nmax}{{\sc NSup}}
\newcommand{\ndi}{{\sc NDisc}}
\newcommand{\nla}{{\sc NLimAvg}}
\newcommand{\nli}{{\sc NLimInf}}
\newcommand{\nls}{{\sc NLimSup}}
\newcommand{\nbw}{{\sc NBW}}

\newcommand{\dmax}{{\sc DSup}}
\newcommand{\ddi}{{\sc DDisc}}
\newcommand{\dla}{{\sc DLimAvg}}
\newcommand{\dli}{{\sc DLimInf}}
\newcommand{\dls}{{\sc DLimSup}}
\newcommand{\dbw}{{\sc DBW}}
\newcommand{\dcw}{{\sc DCW}}

\newcommand{\adi}{{\sc ADisc}}
\newcommand{\ala}{{\sc ALavg}}
\newcommand{\ali}{{\sc ALinf}}

\newcommand{\acw}{{\sc ACW}}

\newcommand{\umax}{{\sc USup}}
\newcommand{\udi}{{\sc UDisc}}

\newcommand{\zmax}{{\sc PosSup}}
\newcommand{\zdi}{{\sc PosDisc}}
\newcommand{\zla}{{\sc PosLimAvg}}
\newcommand{\zli}{{\sc PosLimInf}}
\newcommand{\zls}{{\sc PosLimSup}}
\newcommand{\zbw}{{\sc PosBW}}
\newcommand{\zcw}{{\sc PosCW}}

\newcommand{\asmax}{{\sc AsSup}}
\newcommand{\asdi}{{\sc AsDisc}}
\newcommand{\asla}{{\sc AsLimAvg}}
\newcommand{\asli}{{\sc AsLimInf}}
\newcommand{\asls}{{\sc AsLimSup}}
\newcommand{\asbw}{{\sc AsBW}}
\newcommand{\ascw}{{\sc AsCW}}

\newcommand{\ndcw}{{\sc 
\raisebox{2.0pt}{\scalebox{0.45}{\begin{tabular}{c}D\\[-3pt]N\end{tabular}}}CW
}}
\newcommand{\ndli}{{\sc 
\raisebox{2.0pt}{\scalebox{0.45}{\begin{tabular}{c}D\\[-3pt]N\end{tabular}}}Linf
}}
\newcommand{\ndmax}{{\sc 
\raisebox{2.0pt}{\scalebox{0.45}{\begin{tabular}{c}D\\[-3pt]N\end{tabular}}}Sup
}}
\newcommand{\nd}{{\sc 
\raisebox{2.0pt}{\scalebox{0.45}{\begin{tabular}{c}D\\[-3pt]N\end{tabular}}}
}}
\newcommand{\anbw}{{\sc 
\raisebox{2.0pt}{\scalebox{0.45}{\begin{tabular}{c}A\\[-3pt]N\end{tabular}}}BW
}}
\newcommand{\anls}{{\sc 
\raisebox{2.0pt}{\scalebox{0.45}{\begin{tabular}{c}A\\[-3pt]N\end{tabular}}}Lsup
}}

\newcommand{\Val}{\mathsf{Val}}

\newcommand{\Max}{\mathsf{Sup}}
\newcommand{\LimSup}{\mathsf{LimSup}}
\newcommand{\LimInf}{\mathsf{LimInf}}
\newcommand{\LimAvg}{\mathsf{LimAvg}}
\newcommand{\Disc}{\mathsf{Disc}}

\def\set#1{\ensuremath{\{#1\}}}
\newcommand{\ok}{\raisebox{0.2em}{$\sqrt{}$}}
\newcommand{\ko}{$\times$}
\newcommand{\weight}{\gamma}
\newcommand{\wt}{\weight}

\newcommand{\ov}{\overline}

\newcommand{\trans}{\delta}
\newcommand{\restr}{\upharpoonright}

\newcommand{\Buchi}{\textrm{B\"uchi}}
\newcommand{\coBuchi}{\textrm{coB\"uchi}}

\newcommand{\Prb}{\mathit{Prb}}

\makeatletter

\begingroup \catcode `|=0 \catcode `[= 1
\catcode`]=2 \catcode `\{=12 \catcode `\}=12
\catcode`\\=12 |gdef|@xcomment#1\end{comment}[|end[comment]]
|endgroup

\def\@comment{\let\do\@makeother \dospecials\catcode`\^^M=10\def\par{}}

\def\begincomment{\@comment\@xcomment}

\makeatother

\newenvironment{comment}{\begincomment}{}

\title{Probabilistic Weighted Automata}

\author{Krishnendu Chatterjee$^{1}$ \and Laurent Doyen$^{2}$ \and Thomas A. Henzinger$^{1,3}$}

\institute{IST Austria (Institute of Science and Technology Austria) \and CNRS, Cachan, France \and
EPFL, Switzerland}

\begin{document}
\pagestyle{plain}
\maketitle

\begin{abstract}
Nondeterministic weighted automata are finite automata with numerical
weights on transitions.  They define quantitative languages $L$ that
assign to each word $w$ a real number~$L(w)$.  The value of an
infinite word $w$ is computed as the maximal value of all runs
over~$w$, and the value of a run as the maximum, limsup, liminf, limit
average, or discounted sum of the transition weights.  We introduce
probabilistic weighted automata, in which the transitions are chosen
in a randomized (rather than nondeterministic) fashion.  Under
almost-sure semantics (resp.\ positive semantics), the value of a word
$w$ is the largest real $v$ such that the runs over $w$ have value at
least $v$ with probability~1 (resp.\ positive probability).

We study the classical questions of automata theory for probabilistic
weighted automata: emptiness and universality, expressiveness, and
closure under various operations on languages.  For quantitative
languages, emptiness and universality are defined as whether the value
of some (resp.\ every) word exceeds a given threshold.  We prove some
of these questions to be decidable, and others undecidable.  Regarding
expressive power, we show that probabilities allow us to define a wide
variety of new classes of quantitative languages, except for
discounted-sum automata, where probabilistic choice is no more
expressive than nondeterminism.  Finally, we give an almost complete
picture of the closure of various classes of probabilistic weighted
automata for the following pointwise operations on quantitative
languages: max, min, sum, and numerical complement.
\end{abstract}

\section{Introduction}

In formal design, specifications describe the set of correct
behaviours of a system.  An implementation satisfies a specification
if all its behaviours are correct.  If we view a behaviour as a word,
then a specification is a language, i.e., a set of words.  Languages
can be specified using finite automata, for which a large number of
results and techniques are known; see \cite{automata,Vardi95}.  We
call them \emph{boolean languages} because a given behaviour is either
good or bad according to the specification.  Boolean languages are
useful to specify functional requirements.

In a generalization of this approach, we consider \emph{quantitative
languages}, where each word is assigned a real number.  The value of a
word can be interpreted as the amount of some resource (e.g., memory
or power) needed to produce it, or as a quality measurement for the
corresponding behaviour \cite{CAHS03,CAFH+06}.
Therefore, quantitative languages are useful to specify non-functional
requirements such as resource constraints, reliability properties, or
levels of quality (such as quality of service).

Quantitative languages can be defined using (nondeterministic)
weighted automata, i.e., finite automata with numerical weights on
transitions \cite{CulikK94,EsikK04}.  In~\cite{CDH08a}, we studied
quantitative languages of infinite words and defined the value of an
infinite word $w$ as the maximal value of all runs of an automaton
over~$w$ (if the automaton is nondeterministic, then there may be many
runs over~$w$).  The value of a run $r$ is a function of the infinite
sequence of weights that appear along~$r$.  There are several natural
functions to consider, such as $\Max$, $\LimSup$,
$\LimInf$, limit average, and discounted sum of weights.  For example,
peak power consumption can be modeled as the maximum of a sequence of
weights representing power usage; energy use, as a discounted sum;
average response time, as a limit average \cite{CCHK+05,CAHS03}.

\begin{figure}[!tb]
\begin{center}
     \subfigure[Low reliability but cheap. \label{fig:aut-net1}]{\input{figures/aut-network1.tex}}
     \subfigure[High reliability but expensive. \label{fig:aut-net2}]{\input{figures/aut-network2.tex}}
\end{center}
\caption{Two specifications of a channel.\label{fig:network}}
\end{figure}
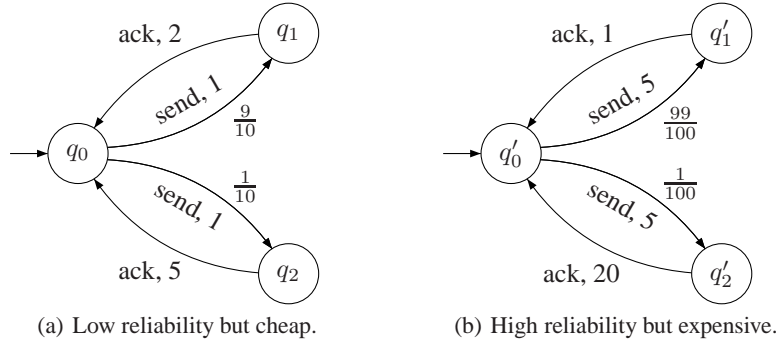

In this paper, we consider \emph{probabilistic} weighted automata as
generators of quantitative languages.  In such automata,
nondeterministic choice is replaced by probability distributions on
successor states.  The value of an infinite word $w$ is defined to be
the maximal value $v$ such that the set of runs over $w$ with value at
least $v$ has either positive probability (\emph{positive semantics}),
or probability~1 (\emph{almost-sure semantics}). 
This simple definition combines in a general model the natural quantitative extensions 
of logics and automata~\cite{AlfaroFS04,DrosteK06,CDH08a}, and the probabilistic 
models of automata for which boolean properties have been well studied~\cite{Rabin63,BlondelC03,BG05}.
Note that the probabilistic B\"uchi and coB\"uchi automata of~\cite{BG05} are 
a special case of probabilistic weighted automata with weights 0 and 1 only 
(and the value of an infinite run computed as $\LimSup$ or $\LimInf$,
respectively). 
While quantitative objectives are standard in the branching-time context of
stochastic games~\cite{Sha53,EM79,FV97,CAHS03,CJH04,Gimbert06},
we are not aware of any model combining probabilities and weights in the 
linear-time context of words and languages, though such a model is very natural
for the specification of quantitative properties. Consider the 
specification of two types of communication channels given in \figurename~\ref{fig:network}.
One has low cost (sending costs $1$ unit) and low reliability
(a failure occurs in 10\% of the case and entails an increased cost for the 
operation), while the second is expensive (sending costs $5$ units),
but the reliability is high (though the cost of a failure is prohibitive).
In the figure, we omit the self-loops with cost $0$ in state $q_0$ and $q'_0$
over {\it ack}, and in $q_1,q_2,q'_1,q'_2$ over {\it send}.
Natural questions can be formulated in this framework, such as whether
the average-cost of every word $w \in \{send,ack\}^\omega$ is really 
smaller in the low-cost channel, or to construct a probabilistic weighted 
automaton that assigns the minimum of the average-cost of the two types
of channels.
In this paper, we attempt a comprehensive study of such fundamental questions, 
about the expressive power, closure properties, and decision problems for 
probabilistic weighted automata.

First, we compare the expressiveness of the various classes of
probabilistic and nondeterministic weighted automata over infinite
words.  For $\LimSup$, $\LimInf$, and limit average, we show that a
wide variety of new classes of quantitative languages can be defined
using probabilities, which are not expressible using nondeterminism.
Our results rely on reachability properties of closed recurrent sets
in Markov chains.  For discounted sum, we show that probabilistic
weighted automata under the positive semantics have the same
expressive power as nondeterministic weighted automata, while under
the almost-sure semantics, they have the same expressive power as
weighted automata with universal branching, where the value of a word
is the minimal (instead of maximal) value of all runs.
The question of whether the positive semantics of weighted limit-average
automata is more expressive than nondeterminism, remains open.

Second, we give an almost complete picture of the closure of
probabilistic weighted automata under the pointwise operations of
maximum, minimum, and sum for quantitative languages. We also define the
\emph{complement} $L^c$ of a quantitative language $L$ by $L^c(w) =
1-L(w)$ for all words~$w$.\footnote{One can define $L^c(w) = k-L(w)$ 
for any constant $k$ without changing the results of this paper.}
Note that maximum and minimum are in fact the operation of least upper bound and 
greatest lower bound for the pointwise natural order on quantitative languages
(where $L_1 \leq L_2$ if and only if $L_1(w) \leq L_2(w)$ for all words $w$).
Therefore, they also provide natural generalization of the classical union and
intersection operations of boolean languages.

Note that
closure under max trivially holds for the positive semantics, and
closure under min for the almost-sure semantics.  We also define the
\emph{complement} $L^c$ of a quantitative language $L$ by $L^c(w) =
1-L(w)$ for all words~$w$.  Only $\LimSup$-automata under positive
semantics and $\LimInf$-automata under almost-sure semantics are
closed under all four operations; these results extend corresponding
results for the boolean (i.e., non-quantitative) case~\cite{BG08}.  To
establish the closure properties of limit-average automata, we
characterize the expected limit-average reward of Markov chains.  Our
characterization answers all closure questions except for the language
sum in the case of positive semantics, which we leave open.  Note that
expressiveness results and closure properties are tightly connected.
For instance, because they are closed under max, the
$\LimInf$-automata with positive semantics can be reduced to
$\LimInf$-automata with almost-sure semantics and to
$\LimSup$-automata with positive semantics; and because they are not
closed under complement, the $\LimSup$-automata with almost-sure
semantics and $\LimInf$-automata with positive semantics have
incomparable expressive powers.

Third, we investigate the emptiness and universality problems for
probabilistic weighted automata, which ask to decide if some (resp.\
all) words have a value above a given threshold.  Using our
expressiveness results, as well as \cite{BG08,CDH08b}, we establish 
some decidability and undecidability results for $\Max$,
$\LimSup$, and $\LimInf$ automata; in particular, emptiness and
universality are undecidable for 
$\LimSup$-automata with positive semantics and for 
$\LimInf$-automata with almost-sure semantics,
while the question is open for 
the emptiness of $\LimInf$-automata with positive semantics and for 
the universality of $\LimSup$-automata with almost-sure semantics.
We also prove the decidability of emptiness for probabilistic 
discounted-sum automata with positive semantics,
while the universality problem is as hard as for the nondeterministic 
discounted-sum automata, for which no decidability result is known.
We leave open the case of limit average. 


\section{Definitions}

A quantitative language over a finite alphabet $\Sigma$ is a function $L:\Sigma^{\omega} \to \real$.
A boolean language (or a set of infinite words) is a special case where $L(w) \in \{0,1\}$ for
all words $w \in \Sigma^{\omega}$. Nondeterministic weighted automata define the value of a word 
as the maximal value of a run~\cite{CDH08a}. In this paper, we study probabilistic weighted automata
as generator of quantitative languages.

\smallskip\noindent{\bf Value functions.}
We consider the following value functions $\Val: \rat^\omega \to \real$ to define quantitative languages.
Given an infinite sequence $v=v_0 v_1 \dots$ of rational numbers, define
\begin{itemize}
\item $\Max(v)    = \sup \{v_n \mid n \geq 0\}$;\smallskip
\item $\LimSup(v) = \displaystyle\limsup_{n\to\infty} \ v_n = \lim_{n\to\infty} \sup \{v_i \mid i \geq n\}$;\smallskip
\item $\LimInf(v) = \displaystyle\liminf_{n\to\infty} \ v_n = \lim_{n\to\infty} \inf \{v_i \mid i \geq n\}$;\smallskip
\item $\LimAvg(v) = \displaystyle\liminf_{n\to\infty} \ \frac{1}{n} \sum_{i=0}^{n-1} v_i$;
\item For $0 < \lambda < 1$, $\Disc_{\lambda}(v) = \displaystyle \sum_{i=0}^{\infty} \lambda^i \cdot v_i$;
\end{itemize}

Given a finite set~$S$, a \emph{probabilistic distribution} over~$S$ is
a function $f: S \to [0,1]$ such that $\sum_{s\in S} f(s) = 1$.
We denote by $\D(S)$ the set of all probabilistic distributions over~$S$.

\smallskip\noindent{\bf Probabilistic weighted automata.}
A \emph{probabilistic weighted automaton} is a tuple 
$A=\tuple{Q,\rho_I,\Sigma,\delta,\weight}$ where:
\begin{itemize}
\item $Q$ is a finite set of states;
\item $\rho_I \in \D(Q)$ is the initial distribution;
\item $\Sigma$ is a finite alphabet;
\item $\delta: Q \times \Sigma \to \D(Q)$ is a probabilistic transition function;
\item $\weight: Q \times \Sigma \times Q \to \rat$ is a weight function.
\end{itemize}

We can define a \emph{non-probabilistic} automaton from $A$ by ignoring
the probability values, and saying that $q$ is initial if $\rho_I(q) > 0$,
and $(q,\sigma,q')$ is an edge of $A$ if $\delta(q,\sigma)(q') > 0$.
The automaton~$A$ is \emph{deterministic} if $\rho_I(q_I) = 1$ for some $q_I\in Q$,
and for all $q \in Q$ and $\sigma \in \Sigma$,
there exists $q' \in Q$ such that $\delta(q,\sigma)(q') = 1$. 

A \emph{run} of $A$ over a finite (resp. infinite) word $w=\sigma_1 \sigma_2 \dots$ 
is a finite (resp. infinite) sequence $r = q_0 \sigma_1 q_1 \sigma_2 \dots $ 
of states and letters such that 
\begin{compressEnum}
\itCompress $\rho_I(q_0) > 0$, and
\itCompress $\delta(q_i,\sigma_{i+1},q_{i+1}) > 0$ for all $0 \leq i < \abs{w}$.
\end{compressEnum}
We denote by $\weight(r) = v_0 v_1 \dots$ the sequence of weights that occur in~$r$ 
where $v_i = \weight(q_i,\sigma_{i+1},q_{i+1})$ for all $0 \leq i < \abs{w}$.

The probability of a finite run $r= q_0 \sigma_1 q_1 \sigma_2 \dots \sigma_{k} q_k$ over a finite word 
$w = \sigma_1 \dots \sigma_k$ is 
$\P^A(r) = \rho_I(q_0).\prod_{i=1}^{k} \delta(q_{i-1}, \sigma_{i})(q_i)$. 
For each $w \in \Sigma^{\omega}$, 
the function $\P^A(\cdot)$ defines a unique probability measure over Borel 
sets of (infinite) runs of~$A$ over~$w$. 

Given a value function $\Val: \rat^\omega \to \real$, we say that
the probabilistic $\Val$-automaton $A$ generates 
the quantitative languages defined for all words $w \in \Sigma^{\omega}$ by 
$L^{=1}_A(w) = \sup \{\eta \mid \P^A(\{r \in \Run^A(w) \mid \Val(\weight(r)) \geq \eta\}) = 1 \}$
under the almost-sure semantics, 
and $L^{>0}_A(w) = \sup \{\eta \mid \P^A(\{r \in \Run^A(w) \mid \Val(\weight(r)) \geq \eta\}) > 0\}$
under the positive semantics. 
For non-probabilistic automata, the value of a word is either the maximal
value of the runs (i.e., $L^{\max}_A(w) = \sup \{\Val(\weight(r)) \mid r \in \Run^A(w)\}$
for all $w \in \Sigma^{\omega}$) and the automaton is then called \emph{nondeterministic},
or the minimal value of the runs, and the automaton is then called \emph{universal}.

Note that B\"uchi and coB\"uchi automata (\cite{BG05}) are special cases of respectively 
$\LimSup$- and $\LimInf$-automata, where all weights are either $0$ or $1$. 

\smallskip\noindent{\bf Notations.} 
The first letter in acronyms for classes of automata
can be {\sc N}(ondeterministic), {\sc D}(eterministic), {\sc U}(niversal), 
{\sc Pos} for the language in the positive semantics, or 
{\sc As} for the language in the almost-sure semantics.
We use the notations \nd to denote the classes of automata
whose deterministic version has the
same expressiveness as their nondeterministic version.
When the type of an automaton~$A$ is clear from the context,
we often denote its language simply by~$L_A(\cdot)$ or even~$A(\cdot)$, 
instead of $L^{=1}_A$, $L^{\max}_A$, etc.

\smallskip\noindent{\bf Reducibility.} 
A class $\C$ of weighted automata is \emph{reducible}
to a class $\C'$ of weighted automata if for every $A \in \C$ there exists 
$A' \in \C'$ such that $L_A=L_{A'}$, i.e. $L_{A}(w)=L_{A'}(w)$ for all words~$w$.
Reducibility relationships for (non)deterministic weighted automata
are given in~\cite{CDH08a}.

\smallskip\noindent{\bf Composition.} 
Given two quantitative languages $L, L': \Sigma^{\omega} \to \real$,
we denote by $\max(L,L')$ (resp. $\min(L,L')$ and $L+L'$)    
the quantitative language that assigns $\max\{L(w),L'(w)\}$
(resp. $\min\{L(w),L'(w)\}$ and $L(w) + L'(w)$)     
to each word $w \in \Sigma^{\omega}$. 
The language $1-L$ is called the \emph{complement} of $L$.
The $\max$, $\min$ and complement operators for quantitative languages 
generalize respectively the union, intersection and complement 
operator for boolean languages. 
The closure properties of (non)deterministic weighted automata are given in~\cite{CDH08b}.

\smallskip\noindent{\bf Remark.} 
We sometimes use automata with weight functions $\weight: Q \to \rat$ that 
assign a weight to states instead of transitions. This is a convenient notation
for weighted automata in which from each state, all outgoing transitions 
have the same weight. In pictorial descriptions of probabilistic weighted
automata, the transitions are labeled with probabilities, and states with weights.

\section{Expressive Power of Probabilistic Weighted Automata}

We complete the picture given in~\cite{CDH08a} about reducibility for nondeterministic weighted automata,
by adding the relations with probabilistic automata. 
The results for $\LimInf$, $\LimSup$, and $\LimAvg$ are summarized in \figurename~\ref{fig:reducibility}s, and
for $\Max$- and $\Disc$-automata in Theorems~\ref{theo:zmax-asmax-to-dmax} and~\ref{theo:ndi-zdi}.

\begin{figure}[!tb]
\begin{center}
 \input{figures/reducibility.tex}
 \caption{Reducibility relation. $\C$ is reducible to $\C'$ if $\C \to \C'$. 
Classes that are not connected by an arrow are incomparable. 
Reducibility for the dashed arrow is open.
The $\Disc$-automata are incomparable with the automata in the figure. Their 
reducibility relations are given in Theorem~\ref{theo:ndi-zdi}.
}
 \label{fig:reducibility}
\end{center}
\end{figure}

\subsection{Probabilistic $\Max$-automata}

Like for probabilistic automata over finite words, the quantitative languages definable by
probabilistic and (non)deterministic $\Max$-automata coincide.

\begin{theorem}\label{theo:zmax-asmax-to-dmax}
\zmax\/ and \asmax\/ are reducible to \dmax.
\end{theorem}

\begin{myProof}
It is easy to see that \zmax-automata define
the same language when interpreted as \nmax-automata,
and the same holds for \asmax\/ and \umax.
The result then follows from~\cite[Theorem~9]{CDH08a}.   
\end{myProof}

\subsection{Probabilistic $\LimAvg$-automata}

Many of our results would consider \emph{Markov chains} and 
\emph{closed recurrent states} in Markov chains. 
A \emph{Markov chain} $M=(S,E,\delta)$ consists of a finite set $S$ of states,
a set $E$ of edges, and a probabilistic transition function 
$\delta:S \to \D(S)$. 
For all $s,t \in S$, there is an edge $(s,t) \in E$ iff 
$\delta(s)(t)>0$. 
A \emph{closed recurrent set} $C$ of states in $M$ is a 
bottom strongly connected set of states in the graph $(S,E)$.
We will use the following two key properties of closed recurrent states.
\begin{enumerate}
\item \emph{Property~1.} Given a Markov chain $M$, and a start state $s$, 
with probability~1, the set of closed recurrent states is reached from 
$s$ in finite time.
Hence for any $\epsilon>0$, there exists $k_0$ such that for all $k>k_0$,
for all starting state $s$, the set of closed recurrent states are 
reached with probability at least $1-\epsilon$ in $k$ steps.

\item \emph{Property~2.} 
If a closed recurrent set $C$ is reached, and the limit of  
the expectation of the average weights of $C$ is $\alpha$, 
then for all $\epsilon>0$, there exists a $k_0$ such that for all 
$k > k_0$ the expectation of the average weights for $k$ steps is 
at least $\alpha-\epsilon$.
\end{enumerate}
The above properties are the basic properties of finite state Markov 
chains and closed recurrent states~\cite{KemenyBook}.


\begin{lemma}\label{lemm-long-enough-b}
Let $A$ be a probabilistic weighted automata with alphabet $\Sigma=\set{a,b}$.
Consider the Markov chain arising of $A$ on input $b^\omega$ (we refer 
to this as the $b$-Markov chain) and we use similar notation for the 
$a$-Markov chain.
The following assertions hold:
\begin{enumerate}
\item If for all closed recurrent sets $C$ in the $b$-Markov chain, the 
(expected) limit-average value (in probabilistic sense) is at least~1, 
then there exists $j$ such that for all closed recurrent sets arising 
of $A$ on input $(b^j \cdot a)^\omega$ the expected limit-average reward is 
positive.

\item If for all closed recurrent sets $C$ in the $b$-Markov chain, the 
(expected) limit-average value (in probabilistic sense) is at most~0, 
then there exists $j$ such that for all closed recurrent sets arising 
of $A$ on input $(b^j \cdot a)^\omega$ the expected limit-average reward is 
strictly less than~1.

\item If for all closed recurrent sets $C$ in the $b$-Markov chain, the 
(expected) limit-average value (in probabilistic sense) is at most~0,
and if for all closed recurrent sets $C$ in the $a$-Markov chain, the 
(expected) limit-average value (in probabilistic sense) is at most~0,
then there exists $j$ such that for all closed recurrent sets arising 
of $A$ on input $(b^j \cdot a^j)^\omega$ the expected limit-average reward is 
strictly less than~1/2.
\end{enumerate}
\end{lemma}
\begin{myProof} 
We present the proof in three parts.
\begin{enumerate}
\item Let $\beta$ be the maximum absolute value of the weights of $A$. 
From any state $s \in A$, there is a path of length at most $n$ 
to a closed recurrent set $C$ in the $b$-Markov chain, where $n$ is 
the number of states of $A$.
Hence if we choose $j>n$, then any closed recurrent set in the 
Markov chain arising on the input $(b^j \cdot a)^\omega$ contains 
closed recurrent sets of the $b$-Markov chain.
For $\epsilon>0$, there exists $k_{\epsilon}$ such that from any state 
$s \in A$, for all $k>k_{\epsilon}$, on input $b^k$ from $s$, the closed 
recurrent sets of the $b$-Markov chain is reached with probability
at least $1-\epsilon$ (by property~1 for Markov chains).
If all closed recurrent sets in the $b$-Markov chain have expected 
limit-average value at least~1, then (by property~2 for Markov chains) 
for all $\epsilon>0$, there 
exists $l_\epsilon$ such that for all $l > l_{\epsilon}$, from 
all states $s$ of a closed recurrent set on the input $b^l$ the 
expected average of the weights is at least $1-\epsilon$, (i.e.,
expected sum of the weights is $l-l\cdot\epsilon$).
Consider $0<\epsilon \leq \min\set{1/4, 1/(20\cdot\beta)}$, 
we choose $j=k+ l$, where $k=k_\epsilon >0 $ and  $l> \max\set{l_\epsilon,k}$.
Observe that by our choice $j+1 \leq 2l$.
Consider  a closed recurrent set in the Markov chain on $(b^j \cdot a)^\omega$
and we obtain a lower bound on the expected average reward as follows:
with probability $1-\epsilon$ the closed recurrent set of the 
$b$-Markov chain is reached within $k$~steps, and then in the next 
$l$ steps at the expected sum of the weights is at least $l-l\cdot 
\epsilon$, and since the worst case weight is $-\beta$ we obtain 
the following bound on the expected sum of the rewards 
\[
(1-\epsilon)\cdot(l-l\cdot\epsilon) - \epsilon \cdot \beta\cdot (j+1) 
\geq \frac{l}{2} - \frac{l}{10} 
=\frac{2l}{5} 
\]
Hence the expected average reward is at least $1/5$ and hence positive.

\item The proof is similar to the previous result.

\item The proof is also similar to the first result. The only difference is 
that we use a long enough sequence of $b^j$ such that with high probability 
a closed recurrent set in the $b$-Markov chain is reached and 
then stay long enough in the  closed recurrent set to approach the expected
sum of rewards to~0, and then present a long enough sequence of $a^j$ such 
that with high probability a closed recurrent set in the $a$-Markov chain is 
reached  and then stay long enough in the  closed recurrent set to 
approach the expected sum of rewards to~0. 
The calculation is similar to the first part of the proof.
\end{enumerate}
Thus we obtain the desired result.
\qed
\end{myProof}

\begin{figure}[t]
   \begin{center}
      \input{figures/aut-ZLa.tex}
   \end{center}
  \caption{A \zla\/ for Lemma~\ref{lemm-fin-lang-avg}.}
  \label{figure:aut-zla}
\end{figure} 

We consider the alphabet $\Sigma$ consisting of letters $a$ and $b$,
i.e., $\Sigma=\set{a,b}$.
We define the language $L_F$ of finitely many $a$'s, i.e., for an infinite 
word $w$ if $w$ consists of infinitely many $a$'s, then $L_F(w)=0$, otherwise
$L_F(w)=1$. We also consider the language $L_I$ of words with infinitely many 
$a$'s (it is the complement of $L_F$).

\begin{lemma}\label{lemm-fin-lang-avg}
Consider the language $L_F$ of finitely many $a$'s. 
The following assertions hold.
\begin{enumerate}
\item The language can be expressed as a \nla.
\item The language can be expressed as a \zla. 
\item The language cannot be expressed as \asla. 
\end{enumerate}
\end{lemma}  
\begin{myProof} We present the three parts of the proof. 
\begin{enumerate}
\item The result follows from the results of~\cite[Theorem~12]{CDH08a} where the
explicit construction of a \nla\/ to express $L_F$ is presented. 

\item A \zla\/ automaton $A$ to express $L_F$ is as follows (see \figurename~\ref{figure:aut-zla}): 
\begin{enumerate}
\item \emph{States and weight function.} 
The set of states of the automaton is $\set{q_0,q_1,\mathit{sink}}$, 
with $q_0$ as the starting state.
The weight function $\wt$ is as follows: $\wt(q_0)=\wt(\mathit{sink})=0$ 
and $\wt(q_1)=1$.
\item \emph{Transition function.} The probabilistic transition function is as follows: 
\begin{verse}
	(i)~from $q_0$, given $a$ or $b$, the next states are $q_0,q_1$, 
	each with probability 1/2; \\
	(ii)~from $q_1$ given $b$ the next state is $q_1$ with probability 1,
	and from $q_1$ given $a$ the next state is $\mathit{sink}$ with probability~1; and \\
	(iii)~from $\mathit{sink}$ state the next state is $\mathit{sink}$ with probability~1 on both $a$ and $b$.
	(it is an absorbing state). 
\end{verse}

\end{enumerate}	
  Given the automaton $A$ consider any word $w$ with infinitely many $a$'s then, 
  the automata reaches sink state in finite time with probability~1, and hence $A(w)=0$.
  For a word $w$ with finitely many $a$'s, let $k$ be the last position
  that an $a$ appears. Then with probability $1/2^k$, after $k$ steps, the automaton 
  only visits the state $q_1$ and hence $A(w)=1$. 
  Hence there is a \zla\/ for $L_F$.

\item We show that $L_F$ cannot be expressed as an \asla. 
Consider an \asla\/ automaton $A$. 
Consider the Markov chain that arises from $A$ if the input is only $b$ 
(i.e., on $b^\omega$), we refer to it as the $b$-Markov chain. 
If there is a closed recurrent set $C$ that can be reached from the starting state 
(reached by any sequence of $a$ and $b$'s), then the limit-average reward 
(in probabilistic sense) in $C$ must be at least~1 (otherwise, if there is a closed 
recurrent set $C$ with limit-average reward less than~1, we can construct a finite word
$w$ that with positive probability will reach $C$, and then follow $w$ by $b^\omega$ and 
we will have $A(w \cdot b^\omega)< 1$).
Hence any closed recurrent set on the $b$-Markov chain has limit-average reward at 
least~1 and by Lemma~\ref{lemm-long-enough-b} there exists $j$ such that 
the $A((b^j \cdot a)^\omega)>0$.
Hence it follows that $A$ cannot express $L_F$.

\end{enumerate}
Hence the result follows. 
\qed
\end{myProof}


\begin{lemma}\label{lemm-inf-lang-avg}
Consider the language $L_I$ of infinitely many $a$'s. 
The following assertions hold.
\begin{enumerate}
\item The language cannot be expressed as an \nla.
\item The language cannot be expressed as a \zla. 
\item The language can be expressed as \asla. 
\end{enumerate}
\end{lemma}  

\begin{figure}[t]
   \begin{center}
      \input{figures/aut-AsLa.tex}
   \end{center}
  \caption{An \asla\/ for Lemma~\ref{lemm-inf-lang-avg}.}
  \label{figure:aut-asla}
\end{figure} 

\begin{myProof}
We present the three parts of the proof.
\begin{enumerate}

\item It was shown in the proof of~\cite[Theorem~13]{CDH08a} that \nla\/ cannot express $L_I$.

\item We show that $L_I$ is not expressible by a \zla. Consider a \zla\/ $A$ and 
 consider the $b$-Markov chain arising from $A$ under the input $b^\omega$. 
 All closed recurrent sets $C$ reachable from the starting state must have 
 the limit-average value at most $0$ (otherwise we can construct an word $w$ with 
 finitely many $a$'s such that $A(w)>0$). Since all closed recurrent set in the $b$-Markov chain
 has limit-average reward that is~0, using Lemma~\ref{lemm-long-enough-b} 
 we can construct a word  $w=(b^j\cdot a)^\omega$, for a large enough $j$, such that $A(w)<1$.
 Hence the result follows.

\item We now show that $L_I$ is expressible as an \asla. The automaton $A$ 
is as follows (see \figurename~\ref{figure:aut-asla}):
\begin{enumerate}
\item \emph{States and weight function.} 
The set of states are $\set{q_0,\mathit{sink}}$ with $q_0$ as the
starting state.
The weight function is as follows: $\wt(q_0)=0$ and $\wt(\mathit{sink})=1$.

\item \emph{Transition function.} The probabilistic transition function is as follows: 
\begin{verse}
 (i)~from $q_0$ given $b$ the next state is $q_0$ with probability~1;\\ 
 (ii)~at $q_0$ given $a$ the next states are $q_0$ and $\mathit{sink}$
  each with probability 1/2;
 (iii)~the $\mathit{sink}$ state is an absorbing state.
\end{verse}
\end{enumerate}
  Consider a word $w$ with infinitely many $a$'s, then the probability of reaching the 
  sink state is~1, and hence $A(w)=1$.
  Consider a word $w$ with finitely many $a$'s, and let $k$ be the number of $a$'s, and 
  then with probability $1/2^k$ the automaton always stay in $q_0$, and hence $A(w)=0$.
\end{enumerate}
Hence the result follows.
\qed 
\end{myProof}

\begin{figure}[t]
   \begin{center}
      \input{figures/aut-ZL.tex}
   \end{center}
  \caption{A probabilistic weighted automaton (\zla, \zls, or \zli) for Lemma~\ref{lemm-zvsn}.}
  \label{figure:aut-zl}
\end{figure} 

\begin{lemma}\label{lemm-zvsn}
There exists a language that can be expressed by \zla, \zls\/ and 
\zli, but not by \nla, \nls\/ or \nli.
\end{lemma}
\begin{myProof}
Consider an automaton~$A$  as follows (see \figurename~\ref{figure:aut-zl}): 
\begin{enumerate} 
\item \emph{States and weight function.} 
The set of states are $\set{q_0,q_1,\mathit{sink}}$ with $q_0$ as the 
starting state.
The weight function is as follows: $\wt(q_0)=\wt(q_1)=1$ and 
$\wt(\mathit{sink})=0$.
\item \emph{Transition function.} 
The probabilistic transition is as follows:
\begin{verse}
(i)~from $q_0$ if the input letter is $a$, then the next states are 
$q_0$ and $q_1$ with probability~1/2; \\
(ii)~from $q_0$ if the input letter is $b$, then the next state is $\mathit{sink}$
with probability~1; \\
(iii)~from $q_1$, if the input letter is $b$, then the next state is $q_0$ with 
probability~1; \\
(iv)~from $q_1$, if the input letter is $a$, then the next state is $q_1$ with 
probability~1; and \\
(v)~the state $\mathit{sink}$ is an absorbing state.
\end{verse}
\end{enumerate}
If we consider the automaton $A$, and interpret it  as a 
\zla, \zls, or \zli, then it accepts the following language:
\[
L_z=\set{a^{k_1} b a^{k_2} b a^{k_3} b \ldots \mid k_1, k_2, \cdots \in \nat_{\geq 1} \cdot
\prod_{i=1}^\infty (1 -\frac{1}{2^{k_i}}) >0} \cup (a \cup b)^* \cdot a^\omega;
\]
i.e., $A(w) =1$ if $w \in L_z$ and $A(w)=0$ if $w\not\in  L_z$:
the above claim follows easily from the argument following Lemma~5 of~\cite{BG05}.
We now show that $L_z$ cannot be expressed as \nla, \nls\/ or \nli.
Consider a non-deterministic automaton $A$. 
Suppose there is a cycle $C$ in $A$ such that average of the rewards in 
$C$ is positive, and $C$ is formed by a word that contains a $b$.
If no such cycle exists, then clearly $A$ cannot express $L_z$ as there 
exists word for which $L_z(w)=1$ such that $w$ contains infinitely many $b$'s.
Consider a cycle $C$ such that average of the rewards is positive, and 
let the cycle be formed by a finite word $w_C=a_0 a_1 \ldots a_n$ and 
there must exist at least one index $0 \leq i \leq n$ such that $a_i=b$.
Hence the word can be expressed as  $w_C= a^{j_1} b a^{j_2} b \ldots a^{j_k} b$,
and hence there exists a finite word $w_R$ (that reaches the cycle) such
that $A(w_R \cdot w_C^\omega)>0$. 
This contradicts that $A$ is an automaton to express $L_z$ as 
$L_z(w_R \cdot w_C^\omega)=0$.
Simply exchanging the average reward of the cycle by the maximum reward 
(resp. minimum reward) shows that $L_z$ is not expressible by a 
\nls\/ (resp. \nli).
\qed
\end{myProof}

The next theorem summarizes the results for limit-average automata obtained in this section.

\begin{theorem}
\asla\/ is incomparable in expressive power with \zla\/ and \nla, and
\nla\/ cannot express all languages expressible by \zla.
\end{theorem}

\medskip\noindent{\bf Open question.} Whether \nla\/ is reducible to \zla\/ or 
\nla\/ is incomparable to \zla\/ (i.e., there is a language expressible by
\nla\/ but not by a \zla) remains open.

\subsection{Probabilistic $\LimInf$-automata}

\begin{lemma}\label{lemm-liminf-ntozas}
\nli\/ is reducible to both \asli\/ and \zli.
\end{lemma}
\begin{myProof}
It was shown in~\cite{CDH08a} that \nli\/ is reducible to 
\dli.
Since \dli\/ are special cases of \asli\/ and \zli\/ the
result follows.
\qed
\end{myProof}

\begin{lemma}\label{lemm-liminf-asvz}
The language $L_I$ is expressible by an \asli, but cannot be 
expressed as a \nli\/ or a \zli. 
\end{lemma}
\begin{myProof}
It was shown in~\cite{CDH08a} that the language $L_I$ is not 
expressible by \nli.
If we consider the automaton $A$ of Lemma~\ref{lemm-inf-lang-avg}
and interpret it as an \asli, then the automaton $A$ expresses
the language $L_I$.
The proof of the fact that \zli\/ cannot express $L_I$ is similar to the 
the proof of Lemma~\ref{lemm-inf-lang-avg} (part(2)) and instead of the
average reward of the closed recurrent set $C$, we need to consider the 
minimum reward of the closed recurrent set $C$.
\qed
\end{myProof}

\begin{lemma}\label{lemm-inf-ztoas}
\zli\/  is reducible to \asli.
\end{lemma}
\begin{myProof}
Let $A$ be a \zli\/ and we construct a \asli\/ $B$ such 
that $B$ is equivalent to $A$. 
Let $V$ be the set of weights that appear in $A$ and let 
$v_1$ be the least value in $V$.
For each weight $v \in V$, consider the \zcw\/ $A^v$ 
that is obtained from $A$ by considering all states with weight at 
least $v$ as accepting states.
It follows from the results of~\cite{BG08}  that \zcw\/ is reducible to 
\ascw\/ (it was shown in~\cite{BG08} that \asbw\/ is reducible to \zbw\/ and
it follows easily that dually \zcw\/ is reducible to \ascw).  
Let $D^v$ be an \ascw\/ that is equivalent to $A^v$.
We construct a \zli\/ $B^v$ from $D^v$ by assigning weights 
$v$ to the accepting states of $D^v$ and the minimum weight 
$v_1$ to all other states.
Consider a word $w$, and we consider the following cases.
\begin{enumerate}
\item If $A(w)=v$, then for all $v' \in V$ such that 
$v' \leq v$ we have $D^{v'}(w)=1$, (i.e., the \zcw\/ $A^{v'}$ 
and the \ascw\/ $D^{v'}$ accepts $w$).

\item For $v \in V$, if $D^v(w)=1$, then $A(w) \geq v$
\end{enumerate}
It follows from above that 
$A = \max_{v \in V} B^v$.
We will show later that \asli\/ is closed under $\max$ 
(Lemma~\ref{lemm-closed-asliw-zlsw}) and hence 
we can construct an \asli\/ $B$ such that 
$B=\max_{v \in V} B^v$.
Thus the result follows. 
\end{myProof}

\begin{theorem}
We have the following strict inclusion
\begin{center}
\nli\/ $\subsetneq$ \zli\/ $\subsetneq$ \asli
\end{center}
\end{theorem}
\begin{myProof}
The fact that \nli\/ is reducible to \zli\/ follows from 
Lemma~\ref{lemm-liminf-ntozas},
and the fact the \zli\/ is not reducible to \nli\/ follows from
Lemma~\ref{lemm-zvsn}.
The fact that \zli\/ is reducible to \asli\/ follows from 
Lemma~\ref{lemm-inf-ztoas} and the fact that \asli\/ is not 
reducible to \zli\/ follows from Lemma~\ref{lemm-liminf-asvz}.
\qed
\end{myProof}

\subsection{Probabilistic $\LimSup$-automata}

\begin{lemma}\label{lemm-nzlstoasls}
\nls\/ and \zls\/ are not reducible to \asls.
\end{lemma} 
\begin{myProof} 
The language $L_F$ of finitely many $a$'s can be expressed 
as a non-deterministic B\"uchi automata, and hence as a 
\nls.
We will show that \nls\/ is reducible to \zls.
It follows that $L_F$ is expressible as \nls\/ and \zls.
The proof of the fact that \asls\/ cannot express $L_F$ is similar to the 
the proof of Lemma~\ref{lemm-fin-lang-avg} (part(3)) and instead of the
average reward of the closed recurrent set $C$, we need to consider the 
maximum reward of the closed recurrent set $C$.
\qed
\end{myProof}
 

\medskip\noindent{\bf Deterministic in limit \nls.} 
Consider an automaton $A$ that is a \nls. 
Let $v_1 < v_2 < \ldots < v_k$ be the weights that appear in $A$.
We call the automaton $A$ \emph{deterministic in the limit} if for all
states $s$ with weight greater than $v_1$, all states $t$ reachable 
from $s$ are deterministic.

\begin{lemma}
For every \nls\/ $A$, there exists a \nls\/ $B$ that is deterministic 
in the limit and equivalent to $B$. 
\end{lemma}
\begin{myProof}
From the results of~\cite{CY95} it follows that a \nbw\/ $A$ can be 
reduced to an equivalent \nbw\/ $B$ such that $B$ is deterministic in
the limit. 
Let $A$  be a \nls, and let $V$ be the set of weights that appear in $A$.
and let $V= \{v_1,\dots,v_k\}$ with $v_1 < v_2 < \dots < v_k$.
For each $v \in V$, consider the \nbw\/ $A_v$ whose (boolean) language 
is the set of words $w$ such that $L_A(w) \geq v$, by declaring to be accepting
the states with weight at least $v$.
Let $B_v$ be the deterministic in the limit \nbw\/ that is equivalent to 
$A_v$. 
The automaton $B$ that is deterministic in the limit and is equivalent to 
$A$ is obtained as the automaton that by initial non-determinism chooses 
between the $B_v$'s, for $v \in V$.
\qed
\end{myProof}

\begin{lemma}
\nls\/ is reducible to \zls.
\end{lemma}
\begin{myProof}
Given a \nls\/ $A$, consider the \nls\/ $B$ that is deterministic in the 
limit and equivalent to $B$. 
By assigning equal probabilities to all out-going transitions from a state
we obtain a \zls\/ $C$ that is equivalent to $B$ (and hence $A$).
The result follows.
\qed
\end{myProof}

\begin{lemma}
\asls\/ is reducible to \zls.
\end{lemma}
\begin{myProof}
Consider a \asls\/ $A$ and let the weights of $A$ 
be $v_1 < v_2 \ldots < v_l$.
For $1 \leq i \leq l$ consider the \asbw\/ obtained from $A$ 
with the set of state with reward at least $v_i$ as the B\"uchi 
states.
It follows from the results of~\cite{BG08} that 
\asbw\/ is reducible to \zbw.
Let $B_i$ be the \zbw\/ that is equivalent to $A_i$. 
Let $C_i$ be the automaton such that all B\"uchi states of 
$B_i$ is assigned weight $v_i$ and all other states are 
assigned $v_1$.
Consider the automata $C$ that goes with equal probability to 
the starting states of $C_i$, for $1 \leq i \leq l$, and 
we interpret $C$ as a \zls.
Consider a word $w$, and let $A(w)=v_j$ for some $1\leq j \leq l$,
i.e., given $w$, the set of states with reward at least $v_j$ 
is visited infinitely often with probability~1 in $A$.
Hence the \zbw\/ $B_i$ accepts $w$ with positive probability, 
and since $C$ chooses $C_i$ with positive probability, it follows
that given $w$, in $C$ the weight $v_j$ is visited infinitely often
with positive probability, i.e., $C(w) \geq v_j$.
Moreover, given $w$, for all $v_k>v_l$, the set of states with 
weight at least $v_k$ is visited infinitely often with probability~0
in $A$.
Hence for all $k>j$, the automata $B_k$ accepts $w$ with probability~0.
Thus $C(w)< v_k$ for all $v_k>v_j$.
Hence $C(w)=A(w)$ and thus \asls\/ is reducible to \zls.
\qed
\end{myProof}

\begin{lemma}
\asls\/ is not reducible to \nls.
\end{lemma}
\begin{myProof}
It follows from~\cite{BG08} that for $0< \lambda< 1$ the following language $L_\lambda$ 
can be expressed by a \asbw\/ and hence by \asls:
\[
L_{\lambda} =\set{a^{k_1} b a^{k_2} b a^{k_3} b \ldots \mid k_1, k_2, \cdots \in \nat_{\geq 1}. 
\prod_{i=1}^\infty (1 -\lambda^{k_i}) >0}.
\]
It follows from argument similar to Lemma~\ref{lemm-zvsn} that there exists $0 < \lambda <1$ 
such that $L_\lambda$ cannot be expressed by a \nls.
Hence the result follows.
\qed
\end{myProof}

\begin{theorem}
\asls\/ and \nls\/ are incomparable in expressive power, and 
\zls\/ is more expressive than \asls\/ and \nls.
\end{theorem}


\begin{lemma}\label{lemm:zcwtozbw}
\zcw\/ is reducible to \zbw.
\end{lemma}

\begin{myProof}
Let $A=\tuple{Q,q_I,\Sigma,\delta,C}$ be a \zcw\/ with the set 
$C \subseteq Q$ of accepting states.
We construct a \zbw\/  $\ov{A}$ as follows:
\begin{enumerate}
\item The set of states is $Q \cup \ov{Q}$ where $\ov{Q} =\set{\ov{q} \mid q \in Q}$
is a copy of the states in $Q$;

\item $q_I$ is the initial state;

\item The transition function is as follows, for all $\sigma \in \Sigma$: 
\begin{enumerate} 
\item for all states $q,q' \in Q$, we have 
$\ov{\delta}(q,\sigma,q') = \ov{\delta}(q,\sigma,\ov{q'}) = 
\frac{1}{2}\cdot \delta(q,\sigma,q')$, i.e., the state $q'$ 
and its copy $\ov{q'}$ are reached with half of the original 
transition probability;

\item the states $\ov{q} \in \ov{Q}$ such that $q \not\in C$ are
absorbing states (i.e., $\ov{\delta}(\ov{q},\sigma,\ov{q}) = 1$);

\item for all states $q \in C$ and $q' \in Q$,
we have $\ov{\delta}(\ov{q},\sigma,\ov{q'}) = \delta(q,\sigma,q')$, i.e., 
the transition function in the copy automaton follows that of $A$ 
for states that are copy of the accepting states.
 
\end{enumerate}
\item The set of accepting states is $\ov{C} = \set{\ov{q} \in \ov{Q} 
\mid q \in C}$.
\end{enumerate}
We now show that the language of the \zcw\/ $A$ and the language of 
\zbw\/ $\ov{A}$ coincides.
Consider a word $w$ such that $A(w)=1$.
Let $\alpha$ be the probability that given the word $w$ 
eventually always states in $C$ are visited in $A$, and since $A(w)=1$
we have $\alpha>0$.
In other words, as limit $k$ tends to $\infty$, the probability that 
after $k$ steps only states in $C$ are visited is $\alpha$.
Hence there exists $k_0$ such that the probability that after $k_0$ 
steps only states in $C$ are visited is at least $\frac{\alpha}{2}$.
In the automaton $\ov{A}$, the probability to reach states of $\ov{Q}$
after $k_0$ steps has probability $p = 1 - \frac{1}{2^{k_0}} > 0$.
Hence with positive probability (at least 
$p \cdot \frac{\alpha}{2}$) 
the automaton visits infinitely often the states of $\ov{C}$, and hence
$\ov{A}(w)=1$.
Observe that since every state in $\ov{Q} \setminus \ov{C}$ 
is absorbing and non-accepting), it follows that if we consider an accepting 
run $\ov{A}$, then the run must eventually always visits 
states in $\ov{C}$ (i.e., the copy of the accepting states $C$).
Hence it follows that for a given word $w$, if $\ov{A}(w)=1$, then with 
positive probability eventually always states in $C$ are visited in 
$A$. Thus $A(w)=1$, and the result follows.
\end{myProof}

\begin{lemma}
\zli\/ is reducible to \zls, and \asls\/ is reducible to \asli.
\end{lemma}

\begin{myProof}
We present the proof that \zli\/ is reducible to \zls, the 
other proof being similar.
Let $A$  be a \zli, and let $V$ be the set of weights that appear in $A$.
For each $v \in V$, it is easy to construct a \zcw\/ $A_v$ whose (boolean) language 
is the set of words $w$ such that $L_A(w) \geq v$, by declaring to be accepting
the states with weight at least $v$.
We then construct for each $v \in V$ a \zbw\/ $\ov{A}_v$ 
that accepts the language of $A_v$ (such 
a \zbw\/ can be constructed by Lemma~\ref{lemm:zcwtozbw}).
Finally, assuming that $V= \{v_1,\dots,v_n\}$ with $v_1 < v_2 < \dots < v_n$,
we construct the \zls\/ $B_i$ for $i=1,2,\dots,n$ where $B_i$ is obtained
from $\ov{A}_{v_i}$ by assigning weight $v_{i}$ to each accepting states,
and $v_1$ to all the other states.
The \zls\/ that expresses the language of $A$ is 
$\max_{i=1,2\ldots,n} B_i$ and since \zls\/ is closed under $\max$ (see Lemma~\ref{lem:closure-max-one}), the 
result follows.
\end{myProof}


\begin{lemma}\label{lemm:comparing-linf-lsup}
\asli\/ and \zls\/ are reducible to each other; \asls\/ and \zli\/ have incomparable expressive power.
\end{lemma}

\begin{myProof}
This result is an easy consequence of the fact that an automaton interpreted as \asli\/ defines
the complement of the language of the same automaton interpreted as \zls\/ (and similarly for \asls\/ and \zli),
and from the fact that \asli\/ and \zls\/ are closed under complement, while \asls\/ and \zli\/ are not (see Lemma~\ref{lem:asli-zls-complement} and~\ref{lem:asls-zli-complement}).
\end{myProof}

\subsection{Probabilistic $\Disc$-automata}

For probabilistic discounted-sum automata, 
the following result establishes equivalence of the nondeterministic
and the positive semantics, and the equivalence of the universal and
the almost-sure semantics.

\begin{theorem}\label{theo:ndi-zdi}
The following assertions hold:
(a)~\ndi\/ and \zdi\/ are reducible to each other; 
(b)~\udi\/ and \asdi\/ are reducible to each other.
\end{theorem}

\begin{myProof}
(a)~We first prove that \ndi\/ is reducible to \zdi.
Let $A=\tuple{Q,\rho_I,\Sigma,\delta_A,\weight}$ be a \ndi, and let $v_{\min},v_{\max}$
be its minimal and maximal weights respectively. 
Consider the \zdi\/ $B=\tuple{Q,\rho_I,\Sigma,\delta_B,\weight}$
where $\delta_B(q,\sigma)$ is the uniform distribution over the 
set of states $q'$ such that $(q,\sigma,\{q'\}) \in \delta_A$.
Let $r = q_0 \sigma_1 q_1 \sigma_2 \dots $ be a run of $A$ (over $w=\sigma_1 \sigma_2 \dots$)
with value $\eta$. For all $\epsilon > 0$,
we show that $\P^B(\{r \in \Run^B(w) \mid \Val(\weight(r)) \geq \eta - \epsilon\}) > 0\}$.
Let $n \in \nat$ such that $\frac{\lambda^n}{1-\lambda}\cdot (v_{\max} - v_{\min}) \leq \epsilon$,
and let $r_n = q_0 \sigma_1 q_1 \sigma_2 \dots \sigma_n q_n$. 
The discounted sum of the weights in $r_n$ is at least $\eta - \frac{\lambda^n}{1-\lambda}\cdot (v_{\max})$.
The probability of the set of runs over $w$ that are continuations of $r_n$
is positive, and the value of all these runs is at least $\eta - \frac{\lambda^n}{1-\lambda}\cdot (v_{\max} - v_{\min})$,
and therefore at least $\eta - \epsilon$.
This shows that $L_B(w) \geq \eta$, and thus $L_B(w) \geq L_A(w)$. 
Note that $L_B(w) \leq L_A(w)$ since there is no run in $A$ (nor in $B$) 
over $w$ with value greater than $L_A(w)$. Hence $L_B = L_A$.

Now, we prove that \zdi\/ is reducible to \ndi. Given a
\zdi\/ $B=\tuple{Q,\rho_I,\Sigma,\delta_B,\weight}$,
we construct a \ndi\/ $A=\tuple{Q,\rho_I,\Sigma,\delta_A,\weight}$
where $(q,\sigma,\{q'\}) \in \delta_A$ if and only if
$\delta_B(q,\sigma)(q') >0$, for all $q,q' \in Q$, $\sigma \in \Sigma$.
By analogous arguments as in the first part of the proof, it is easy 
to see that $L_B = L_A$.

(b)~It is easy to see that the complement of the quantitative language defined
by a \udi\/ (resp. \asdi) can be defined by a \ndi\/ (resp. \zdi).
Then, the result follows from Part~$a)$ (essentially,
given a \udi, we obtain easily an \ndi\/ for the complement, then
an equivalent \zdi, and finally a \asdi\/ for the complement of the complement, i.e., the original quantitative language).
\qed
\end{myProof}

Note that a by-product of this proof is that the language of a \zdi\/ does not depend
on the precise values of the probabilities, but only on whether they are positive or not.

\section{Closure Properties of Probabilistic Weighted Automata}


We consider the closure properties of the probabilistic weighted automata
under the operations $\max$, $\min$, complement, and sum. The results are
presented in Table~\ref{tab:closure-properties}.

\begin{table}[t]
\begin{center}
\begin{tabular}{|c|l|*{4}{c|}c|*{2}{c|}}
\cline{1-6}\cline{8-9}
&       & $\max$ & $\min$ & comp. & sum & &   \,emptiness\,  & universality\\
\cline{1-6}\cline{8-9}
\multirow{5}{*}{\rotatebox{90}{{\scriptsize $>0$}}} & \zmax   &  \ok & \ok  & \ko   & \ok & &  \ok  &  \ok \\
\cline{2-6}\cline{8-9}
& \zls   &  \ok & \ok  & \ok   & \ok & & \ko &  \ko \\
\cline{2-6}\cline{8-9}
& \zli   &  \ok & \ok  & \ko   & \ok & & \ok &  \ok  \\
\cline{2-6}\cline{8-9}
& \zla\,   &  \ok & \ko  & \ko   &  ?   & & ?      & ?     \\
\cline{2-6}\cline{8-9}
& \zdi   &  \ok & \ko  & \ko   & \ok & & \ok  &  ? (1)  \\
\cline{1-6}\cline{8-9}\\[-8pt]
\cline{1-6}\cline{8-9}
\multirow{5}{*}{\rotatebox{90}{{\scriptsize almost-sure}\ }} & \asmax &  \ok & \ok  & \ko   & \ok & &   \ok &  \ok \\
\cline{2-6}\cline{8-9}
& \asls  &  \ok & \ok  & \ko   & \ok & &  \ok  &  \ok\\
\cline{2-6}\cline{8-9}
& \asli  &  \ok & \ok  & \ok   & \ok & &  \ko  &  \ko\\
\cline{2-6}\cline{8-9}
& \asla  &  \ko & \ok  & \ko   & \ko & &  ?     & ? \\
\cline{2-6}\cline{8-9}
& \asdi  &  \ko & \ok  & \ko   & \ok & &  ? (1)  &  \ok \\
\cline{1-6}\cline{8-9}
\multicolumn{9}{c}{{\large \strut}\parbox[t]{82mm}{The universality problem for \ndi\/ can be reduced to~(1). It is not known whether this problem is decidable.}}
\end{tabular}
\end{center}
\caption{Closure properties and decidability of the emptiness and universality problems.\label{tab:closure-properties}}
\end{table}

\subsection{Closure under $\max$ and $\min$}

\begin{lemma}[Closure by initial non-determinism]\label{lem:closure-max-one}
\zls, \zli\/ and \zla\/ is closed under $\max$; and 
\asls, \asli\/ and \asla\/ is closed under $\min$.
\end{lemma}
\begin{myProof}
Given two automata $A_1$ and $A_2$ consider the automata $A$ 
obtained by initial non-deterministic choice of $A_1$ and $A_2$.
Formally, let $q_1$ and $q_2$ be the initial states of $A_1$ and 
$A_2$, respectively, then in $A$ we add an initial state $q_0$ and
the transition from $q_0$ is as follows:
for $\sigma \in \Sigma$, consider the set $Q_{\sigma} = 
\set{q \in Q_1 \cup Q_2 \mid \trans_1(q_1,\sigma)(q)>0 \text{ or } 
\trans_2(q_2,\sigma)(q) >0}$.
From $q_0$, for input letter $\sigma$, the successors are from $Q_{\sigma}$ 
each with probability $1/|Q_{\sigma}|$.
If $A_1$ and $A_2$ are \zls\/ (resp. \zli, \zla), then 
$A$ is a \zls\/ (resp. \zli, \zla) such that $A=\max\set{A_1,A_2}$.
Similarly, if $A_1$ and $A_2$ are \asls\/ (resp. \asli, \asla), then 
$A$ is a \asls\/ (resp. \asli, \asla) such that $A=\min\set{A_1,A_2}$.
\qed
\end{myProof}

\begin{lemma}[Closure by synchronized product]\label{lem:closure-max-two}
\asls\/ is closed under $\max$ and 
\zli\/ is closed under $\min$.
\end{lemma}
\begin{myProof}
We present the proof that \asls\/ is closed under $\max$. 
Let $A_1$ and $A_2$ be two probabilistic weighted automata with 
weight function $\wt_1$ and $\wt_2$, respectively.
Let $A$ be the usual synchronized product of $A_1$ and $A_2$ with 
weight function $\wt$ such that 
$\wt((s_1,s_2))=\max\set{\wt_1(s_1),\wt_2(s_2)}$.
Given a path $\pi=((s_0^1,s_0^2),(s_1^1,s_1^2), \ldots)$ in $A$ we denote by
$\pi \restr 1$ the path in $A_1$ that is the projection of the first component 
of $\pi$ and we use similar notation for $\pi \restr 2$.
Consider a word $w$, let $\max\set{A_1(w),A_2(w)}=v$. 
We consider the following two cases to show that $A(w)=v$.
\begin{enumerate}

\item W.l.o.g. let the maximum be achieved by $A_1$, i.e., $A_1(w)=v$. 
Let $B_i^v$ be the set of states $s_i$ in $A_i$ such that weight of $s_i$ 
is at least $v$.
Since $A_1(w)=v$, given the word $w$, in $A_1$ the event $\Buchi(B_1^v)$ holds 
with probability~1.
Consider the following set of paths in $A$
\[
\Pi^v =\set{\pi \mid (\pi \restr 1) \in \Buchi(B_1^v) }.
\]
Since given $w$, the event $\Buchi(B_1^v)$ holds with probability~1 in $A_1$,
it follows that given $w$, the event $\Pi^v$ holds with probability~1 in $A$.
The $\wt$ function ensures that every path $\pi \in \Pi^v$ visits weights 
of value at least $v$ infinitely often.
Hence $A(w) \geq v$.

\item Consider a weight value $v' > v$. 
Let $C_i^v$ be the set of states $s_i$ in $A_i$ such that the weight of 
$s_i$ is less than $v'$.
Given the word $w$, since $A_i(w) < v'$, it follows that 
probability of the event $\coBuchi(C_i^v)$ in $A_i$, given the word $w$, 
is positive.
Hence given the word $w$, the probability of the event 
$\coBuchi(C_1^v \times C_2^v))$ is positive in $A$.
It follows that $A(w) < v'$.
\end{enumerate}
The result follows. 
If $A_1$ and $A_2$ are \zli, and in $A$ we assign weights 
such that every state in $A$ has the minimum weight of its 
component states, and we consider $A$ as a \zli, then 
$A=\min\set{A_1,A_2}$.
The proof is similar to the result for \asls.
\qed
\end{myProof}

\begin{lemma}\label{lemm-closed-asliw-zlsw}
\zls\/ is closed under $\min$ and \asli\/ is 
closed under $\max$.
\end{lemma}
\begin{myProof} 
Let $A_1$  and $A_2$ be two \zls. We construct a \zls\/ $A$ 
such that $A = \min\set{A_1,A_2}$. Let $V_i$ 
be the set of weights that appear in $A_i$ (for $i=1,2$), and let 
$V=V_1 \cup V_2$ and let $v_1$ be the least value in $V$.
For each weight $v \in V_1 \cup V_2 = \{v_1,\dots,v_k\}$, 
consider the \zbw\/ $A_i^v$ that is obtained from $A_i$ by considering 
all states with weight at least $v$ as accepting states.
Since \zbw\/ is closed under intersection(by the results of~\cite{BG05}), 
we can construct a \zbw\/ $A_{12}^v$ that is the intersection of $A_1^v$ and 
$A_2^v$, i.e. $A_{12}^v=A_1^v \cap A_2^v$.
We construct a \zls\/ $B_{12}^v$ from $A_{12}^v$ by assigning weights 
$v$ to the accepting states of $A_{12}^v$ and the minimum weight 
$v_1$ to all other states.
Consider a word $w$, and we consider the following cases.
\begin{enumerate}
\item If $\min\set{A_1(w),A_2(w)}=v$, then for all $v' \in V$ such that 
$v' \leq v$ we have $A_{12}^{v'}(w)=1$, (i.e., the \zbw\/ $A_{12}^{v'}$ 
accepts $w$).

\item If $A_{12}^v(w)=1$, then $A_1(w) \geq v$ and $A_2(w) \geq v$, 
i.e., $\min\set{A_1(w),A_2(w)} \geq v$.
\end{enumerate}
It follows from above that 
$\min\set{A_1,A_2}= \max_{v \in V} B_{12}^v$.
Since \zls\/ is closed under $\max$ (by initial non-determinism), 
it follows that \zls\/ is closed under $\min$.
The proof of closure of \asli\/ under $\max$ is similar.
\qed
\end{myProof}

\noindent The closure properties of $\LimAvg$-automata in the positive semantics rely on the following lemma.

\begin{lemma}\label{lem:limavg-min-comp}
Consider the alphabet $\Sigma=\set{a,b}$, and consider the languages
$L_a$ and $L_b$ that assigns the long-run average number 
of $a$'s and $b$'s, respectively.
Then the following assertions hold.
\begin{enumerate}
\item There is no \zla\/ for the language $L_m=\min\set{L_a,L_b}$.
\item There is no \zla\/ for the language $L^*= 1 -\max\set{L_a,L_b}$.
\end{enumerate}
\end{lemma}

\begin{myProof}
To obtain a contradiction, assume that there exists a \zla\/ 
$A$ (for either $L_m$ or $L^*$). 
We first claim that if we consider the $a$-Markov or the $b$-Markov chain of $A$, 
then there must be either an $a$-closed recurrent set or a 
$b$-closed recurrent set $C$ that is reachable in $A$ 
such that the expected sum of the weights in $C$ is positive.
Otherwise, if for all $a$-closed recurrent sets and $b$-closed recurrent sets we have that 
the expected sum of the weights is zero or negative, then we fool the automaton
as follows.
By Lemma~\ref{lemm-long-enough-b}, it follows that there exists a $j$ such that 
$A( (a^j \cdot b^j) ^\omega) < 1/2$, however, $L_m(w)=L^*(w)=\frac{1}{2}$,
i.e., we have a contradiction.
W.l.o.g., we assume that there is an $a$-closed recurrent set $C$ such that expected 
sum of weights of $C$ is positive.
Then we present the following word $w$: a finite word $w_C$ to reach 
the cycle $C$, followed by $a^\omega$; the answer of the automaton is positive, 
{\it i.e.}, $L_{A}(w)>0$,  while $L_m(w)=L^*(w)=0$.
Hence the result follows.
\end{myProof}

\begin{lemma}\label{lem:zla-min--asla-max}
\zla\/ is not closed under $\min$ and \asla\/ is not closed under $\max$.
\end{lemma}
\begin{myProof}
The result for \zla\/ follows from Lemma~\ref{lem:limavg-min-comp}.
We now show that \asla\/ is not closed under $\max$.
Consider the alphabet $\Sigma=\set{a,b}$ and the quantitative languages
$L_a$ and $L_b$ that assign the value of long-run average
number of $a$'s and $b$'s, respectively.
There exists \dla\/ (and hence \asla) for $L_a$ and $L_b$.
We show that $L_m=\max(L_a,L_b)$ cannot be expressed by 
an \asla\/. By contradiction, assume that $A$ is an \asla\/
with set of states $Q$ that defines $L_m$.
Consider any $a$-closed recurrent $C$ in $A$. 
The expected limit-average of the weights 
of the recurrent set must be~1, as if we consider the 
word $w^*=w_C \cdot a^\omega$ where $w_C$
is a finite word to reach $C$, the value
of $w^*$ in $L_m$ is $1$. 
Hence, the limit-average of the weights of all the reachable 
$a$-closed recurrent set $C$ in $A$ is~1.

Given $\epsilon>0$, there exists $j_\epsilon$ such that the following 
properties hold:
\begin{enumerate}
\item from any state of $A$, given the word $a^{j_\epsilon}$ with 
probability $1-\epsilon$ an $a$-closed recurrent set is reached 
(by property~1 for Markov chains);

\item once an  $a$-closed recurrent set is reached, given the word 
$a^{j_\epsilon}$, (as a consequence of property~2 for Markov chains)
we can show that the following properties hold:
(a)~the expected average of the weights is at least $j_\epsilon\cdot(1-\epsilon)$, 
and (b)~the probability distribution of the states is with $\epsilon$ of 
the probability distribution of the states for the word 
$a^{2 \cdot j_\epsilon}$ (this holds as the probability distribution 
of states on words $a^j$ 
converges to the probability distribution of states on the word $a^\omega$).

\end{enumerate}
\newcommand{\wh}{\widehat}
Let $\beta>1$ be a number that is greater than the absolute maximum value 
of weights in $A$.
We chose $\epsilon>0$ such that $\epsilon< \frac{1}{40\cdot\beta}$.
Let $j= 2\cdot j_\epsilon$ (such that $j_\epsilon$ satisfies the 
properties above).
Consider the word $(a^j \cdot b^{3j})^\omega$ and the answer by $A$ must be
$\frac{3}{4}$, as $L_m((a^j \cdot b^{3j})^\omega)=\frac{3}{4}$.
Consider the word $\wh{w}=(a^{2j} \cdot b^{3j})^\omega$ and consider a closed
recurrent set in the Markov chain obtain from $A$ on $\wh{w}$. 
We obtain the following lower bound on the expected limit-average of the 
weights: 
(a)~with probability at least $1-\epsilon$, after $j/2$ steps, $a$-closed
recurrent sets are reached;  
(b)~the expected average of the weights for 
the segment between $a^j$ and $a^{2j}$ is at least $j\cdot(1-\epsilon)$;
and (c)~the difference in probability distribution of the states after 
$a^j$ and $a^{2j}$ is at most $\epsilon$.
Since the limit-average of the weights of $(a^j \cdot b^{3j})^\omega$ is 
$\frac{3}{4}$, the lower bound on the limit-average of the weights is as 
follows
\[
\begin{array}{rcl}
(1-3\cdot\epsilon) \cdot( \frac{3\cdot j + j\cdot(1-\epsilon)  }{5j}) 
-3\cdot\epsilon\cdot\beta
& = & (1-\epsilon)(\frac{4}{5} -\frac{\epsilon}{5}) - 3\cdot\epsilon\cdot\beta \\[1ex] 
& \geq & \frac{4}{5} - \epsilon - 3\cdot \epsilon \cdot \beta \\[1ex]
& \geq & \frac{4}{5} -  4\cdot \epsilon \cdot \beta \\[1ex]
& \geq & \frac{4}{5} -\frac{1}{10} \\[1ex]
& \geq & \frac{7}{10} > \frac{3}{5}.
\end{array}
\]
It follows that $A((a^{2j} \cdot b^{3j})^\omega) > \frac{3}{5}$.
This contradicts that $A$ expresses $L_m$. 
\end{myProof}

\subsection{Closure under complement}

\begin{lemma}\label{lem:asli-zls-complement}
\zls\/ and \asli\/ are closed under complement.
\end{lemma}
\begin{myProof}
We first present the proof for \zls. 
Let $A$  be a \zls, and let $V$ be the set of weights that appear in $A$.
For each $v \in V$, it is easy to construct a \zbw\/ $A_v$ whose (boolean) language 
is the set of words $w$ such that $L_A(w) \geq v$, by declaring to be accepting
the states with weight at least $v$.
We then construct for each $v \in V$ a \zbw\/ $\bar{A}_v$ (with accepting states) 
that accepts the (boolean) complement of the language accepted by $A_v$ (such 
a \zbw\/ can be constructed since \zbw\/ is closed under complementation by the
results of~\cite{BG08}).
Finally, assuming that $V= \{v_1,\dots,v_n\}$ with $v_1 < v_2 < \dots < v_n$,
we construct the \zls\/ $B_i$ for $i=2,\dots,n$ where $B_i$ is obtained
from $\bar{A}_{v_i}$ by assigning weight $-v_{i-1}$ to each accepting states,
and $-v_n$ to all the other states. The complement of $L_A$ is then
$\max\{L_{B_2},\dots,L_{B_n}\}$ which is accepted by a \zls\/ 
(since \zls\/ is closed under $\max$).
The result for \asli\/ is similar and it uses the closure of 
\ascw\/ under complementation which can be easily proved from the 
closure under complementation of \zbw.
\end{myProof}

\begin{lemma}\label{lem:asls-zli-complement}
\asls\/ and \zli\/ are not closed under complement.
\end{lemma}

\begin{proof}
It follows from Lemma~\ref{lemm-nzlstoasls}  that the language $L_F$ of finitely 
many $a$'s is not expressible by an \asls, whereas the complement $L_I$ of 
infinitely many $a$'s is expressible as a \dbw\/ and hence as a \asls.
It follows from Lemma~\ref{lemm-liminf-asvz} that language $L_I$ is not expressible as an 
\zli, whereas its complement $L_F$ is expressible by a \dcw\/ and hence 
a \zli.
\qed
\end{proof}

\begin{lemma}
\zla\/ and \asla\/ are not closed under complement.
\end{lemma}
\begin{myProof}
The fact that \zla\/ is not closed under complement follows from 
Lemma~\ref{lem:limavg-min-comp}.
We now show that \asla\/ is not closed under complement.
Consider the \dla\/ $A$ over alphabet $\Sigma= \{a,b\}$ that consists of 
a single self-loop state with weight $1$ for $a$ and $0$ for $b$.
Notice that $A(w.a^\omega) = 1$ and $A(w.b^\omega) = 0$ for all 
$w \in \Sigma^*$.
To obtain a contradiction, assume that there exists a \asla\/ $B$ 
such that $B=1-A$. 
For all finite words $w \in \Sigma^*$, let $B(w)$ be the expected average weight 
of the finite run of $B$ over $w$.
Fix $0 < \epsilon < \frac{1}{2}$. For all finite words $w$, there exists
a number $n_w$ such that the average number of $a$'s in $w.b^{n_w}$ is at most 
$\epsilon$, 
and there exists a number $m_w$ such that $B(w.a^{m_w}) \leq \epsilon$
(since $B(w.a^\omega) = 0$). 
Hence, we can construct a word $w = b^{n_1} a^{m_1} b^{n_2} a^{m_2} \dots$ 
such that 
$A(w) \leq \epsilon$ and $B(w) \leq \epsilon$.
Since $B = 1-A$, this implies that $1 \leq 2\epsilon$, a contradiction.
\qed
\end{myProof}

\subsection{Closure under sum}

\begin{lemma}\label{lemm:zls-asls-closed-under-sum}
\zls\/ and \asls\/ are closed under sum.
\end{lemma}
\begin{myProof}
Given two \zls\/ (resp. \asls) $A_1$ and $A_2$, 
we construct a \zls\/ (resp. \asls) $A$ for the sum of their 
languages as follows. 
For a pair $(v_1,v_2)$ of weights ($v_i$ in $A_i$, for $i=1,2$),
consider a copy of the synchronized product of $A_1$ and $A_2$. 
We attach a bit $b$ whose range is $\{1,2\}$ to each state to 
remember that we expect $A_b$ to visit the guessed weight $v_b$. 
Whenever this occurs, the bit $b$ is set to $3-b$, and the weight of the state 
is $v_1 + v_2$. All other states ({\it i.e.}
when $b$ is unchanged) have weight 
$\min\{v_1 + v_2 \mid v_1 \in V_1 \land v_2 \in V_2\}$.
Let the automata constructed be $A_{(v_1,v_2)}$.
Then $A =\max_{(v_1,v_2)} A_{(v_1,v_2)}$.
Since \zls\/ (resp. \asls) is closed under $\max$ the result follows. 
\end{myProof}

\begin{lemma}\label{lemm:zli-asli-closed-under-sum}
\zli\/ and \asli\/ are closed under sum.
\end{lemma}
\begin{myProof}
Given two \zli\/ (resp. \asli) $A_1$ and $A_2$, we construct a \zli\/
(resp. \asli) $A$ for the sum of their languages as follows. 
For $i=1,2$, let $V_i$ be the set of weights that appear in $A_i$.
Let $v_{\min}=\min\{v_1 + v_2 \mid v_1 \in V_1 \land v_2 \in V_2\}$.
For $v_1 \in V_1$ and $v_2 \in V_2$, for $i=1,2$, 
consider the \zcw\/ (resp. \ascw) $A_{v_i}$ obtained from $A_i$ by 
making all states with weights at least $v_i$ as accepting states.
Let $A_{(v_1,v_2)}$ be the \zcw\/ (resp. \ascw) such that 
$A_{(v_1,v_2)}= A_{v_1} \cap A_{v_2}$: such an \zcw\/ 
(resp. \ascw) exists since \zcw\/ (resp. \ascw) 
is closed under intersection. 
In other words, for a word $w$ we have 
$A_{(v_1,v_2)}(w) =1$ iff $A_1(w) \geq v_1$ and $A_2(w) \geq v_2$.
Let $\ov{A}_{(v_1,v_2)}$ be the \zli\/ (resp. \asli) obtained 
from $A_{(v_1,v_2)}$ by assigning weight $v_1 + v_2$ to all accepting 
states and weight $v_{\min}$ to all other states.
Then the automaton for the sum of $A_1$ and $A_2$ (denoted as $A_1+A_2$) 
is $\max_{(v_1,v_2) \in V_1 \times V_2} \ov{A}_{(v_1,v_2)}$. 
Since \zli\/ (resp. \asli) is closed under $\max$ the result follows. 
\end{myProof}

\begin{lemma}
\asla\/ is not closed under sum.
\end{lemma}

\begin{myProof}
Consider the alphabet $\Sigma=\set{a,b}$, and consider the \dla-definable languages 
$L_a$ and $L_b$ that assigns to each word $w$ the long-run average number of $a$'s 
and $b$'s in $w$ respectively. 
Let $L_{+}=L_a + L_b$. 
We show that $L_+$ is not expressible by \asla.
Assume towards contradiction that $L_{+}$ is defined by an \asla\/ $A$ with set 
of states $Q$ (we assume w.l.o.g that every state in $Q$ is reachable).
Let $\beta>1$ be greater than the maximum absolute value of the weights in $A$.

First, we claim that from every state $q \in Q$, if we consider the 
automaton $A_q$ with $q$ as starting state then $A_q(a^\omega)=1$:
this follows since if we consider a finite word $w_q$ to reach 
$q$, then $L_+(w_q \cdot a^\omega)=1$ and hence $A(w_q \cdot a^\omega)=1$.
It follows that from any state $q$, as $k$ tends to $\infty$,
the expected average of the weights converges almost-surely to~1.
This implies if we consider the $a$-Markov chain arising from $A$, then
from any state $q$, for all closed recurrent set $C$ of states reachable
from $q$, the expected average of the weights of $C$ is~1.
Hence for every $\gamma>0$ there exists a natural number $k_0^\gamma$ 
such that from any state $q$, for all $k > k_0^\gamma$  
given the word $a^k$ the expected average of the weights is at least $\frac{1}{2}$
with probability $1-\gamma$ (this is because we can chose long enough $k$ 
such that the closed recurrent states are reached with probability $1-\gamma$ 
by property~1 for Markov chains, and then the long enough sequence ensures that 
the expected average approaches~1 by property~2 for Markov chains), and for 
the first $k_0^\gamma$ steps the expected average of the weights is at least $-\beta$.
The same result holds if we consider as input a sequence of $b$'s instead of $a$'s.

Consider the word $w$ generated inductively by the following procedure: 
(a)~$w_0$ is the empty word; 
(b)~we generate $w_{i+1}$ from $w_i$ as follows:
\begin{compressEnum}
\itCompress the sequence of letters added to $w_i$ to obtain $w_{i+1}$ is at 
least $i$;
\itCompress first we  generate a long enough sequence $w_{i+1}'$ of $a$'s 
after $w_i$ such that the average number of $b$'s in $w_i \cdot w_{i+1}'$ falls below $\frac{1}{i}$;
\itCompress then generate a long enough sequence $w_{i+1}''$ of $b$'s such that the average number of 
$a$'s in $w_{i} \cdot w_{i+1}' \cdot w_{i+1}''$ falls below $\frac{1}{i}$;
\itCompress the word $w_{i+1}=w_i \cdot w_{i+1}' \cdot w_{i+1}''$.
\end{compressEnum}
The word $w$ is the limit of these sequences.
For $\gamma>0$, consider $i\geq 6 \cdot k_0^\gamma \cdot \beta$ (where $k_0^\gamma$ 
satisfies the properties described above for $\gamma$). 
By construction for $i > 6\cdot k_0^\gamma \cdot \beta$, the length of $w_i$ is at least 
$6\cdot k_0 \cdot\beta$, 
and hence it follows that in the segment constructed between 
$w_i$ and $w_{i+1}$, for all $|w_i| \leq \ell \leq |w_{i+1}|$ with probability 
at least $1-\gamma$ the expected average of the weights is at least 
\[
\frac{\frac{\ell -k_0^\gamma}{2} - k_0^\gamma\cdot \beta}{\ell} \geq 
\frac{1}{2} -   \frac{2\cdot k_0^\gamma \cdot \beta}{\ell} 
\geq \frac{1}{2} -\frac{1}{3} \geq \frac{1}{6}.
\]
Hence for all $\gamma>0$, the expected average of the weights is 
at least $\frac{1}{6}$ with probability at least $1-\gamma$.
Since this holds for all $\gamma>0$, it follows that the expected
average of the weights is at least $\frac{1}{6}$ almost-surely, 
(i.e., $A(w)\geq \frac{1}{6}$).
We have $L_{a}(w) = L_{b}(w) = 0$ and thus $L_{+}(w) = 0$, while
$A(w) \geq \frac{1}{6}$.
Thus we have a contradiction.
\end{myProof}

\begin{lemma}\label{lemm:zdi-asdi-closed-under-sum}
\zdi\/ and \asdi\/ are closed under sum.
\end{lemma}

\begin{myProof}
The result for \zdi\/ follows from Theorem~\ref{theo:ndi-zdi}
and the fact that \ndi\/ and \udi\/ are closed under sum (which is easy to prove using 
a synchronized product of automata where the weight of a joint transition is the sum
of the weights of the corresponding transitions.
\end{myProof}

\noindent{\bf Open question.} Whether \zla\/ is closed under sum remains open.

\section{Decision Problems for Probabilistic Weighted Automata}

We conclude the paper with some decidability and undecidability results for classical decision problems
about quantitative languages (see Table~\ref{tab:closure-properties}). 
Most of them are direct corollaries of the results in~\cite{BG08}.
Given a weighted automaton $A$ and a rational number $\nu \in \rat$, 
the \emph{quantitative emptiness problem} asks whether there exists a word $w \in \Sigma^{\omega}$
such that $L_{A}(w) \geq \nu$, and the \emph{quantitative universality problem}
asks whether $L_{A}(w) \geq \nu$ for all words $w \in \Sigma^{\omega}$.



\begin{theorem}
The emptiness and universality problems for \zmax\/ and \asmax\/ are decidable.
\end{theorem}

\begin{proof}
By Theorem~\ref{theo:zmax-asmax-to-dmax}, these problems reduce to emptiness
of \dmax\/ which is decidable (\cite[Theorem~1]{CDH08a}).
\end{proof}


The following theorems are trivial corollaries of~\cite[Theorem~2]{BG08}.

\begin{theorem}
The emptiness problem for \zls\/ and the universality problem for \asls\/ are undecidable.
\end{theorem}

It is easy to obtain the following result as a straightforward generalization
of~\cite[Theorem~6]{BG08}.

\begin{theorem}
The emptiness problem for \asls\/ and the universality problem for \zli\/ are decidable.
\end{theorem}

\begin{theorem}
The emptiness problem for \zli\/ and the universality problem for \asls\/ are decidable.
\end{theorem}

\begin{myProof}[(Sketch)]
We sketch the main ideas of the proof that emptiness of $\coBuchi$ automata
in positive semantics is achievable in EXPTIME and with exponential memory. 
The proof extends easily to \zli\/ and to the universality problem for \asls.

Emptiness of $\coBuchi$ automata in positive semantics can be viewed as 
deciding the existence of a blind positive-winning strategy in a stochastic 
game with $\coBuchi$ objective. 
It follows from the results of~\cite{CDH-POMDP} that this problem can be decomposed
into positive winning for safety and reachability objectives.
\end{myProof}

The following result is a particular case of~\cite[Corollary~3]{BG08}.

\begin{theorem}
The emptiness problem for \asli\/ and the universality problem for \zls\/ are undecidable.
\end{theorem}

Finally, by Theorem~\ref{theo:ndi-zdi} and the decidability of emptiness for \ndi,
we get the following result.

\begin{theorem}
The emptiness problem for \zdi\/ and the universality problem for \asdi\/ are decidable.
\end{theorem}

Note that by Theorem~\ref{theo:ndi-zdi}, the universality problem for \ndi\/ (which is not know to be decidable) 
can be reduced to the universality problem for \zdi\/ and to the emptiness problem for \asdi.

\bibliography{main}
\bibliographystyle{plain}
\end{document}

%% file: figures/aut-network1.tex
\unitlength=.8mm
\def\fsize{\normalsize}

\begin{picture}(70,50)(0,0)

{\fsize

\node[Nmarks=i](q0)(18,25){$q_0$}
\node[Nmarks=n](q1)(53,45){$q_1$}
\node[Nmarks=n](q2)(53,5){$q_2$}

\drawedge[ELpos=50, ELside=l, ELdist=-3, syo=1, curvedepth=-6](q0,q1){\rotatebox{30}{send, 1}}
\drawedge[ELpos=65, ELside=r, ELdist=.5, syo=1, curvedepth=-6](q0,q1){$\frac{9}{10}$}   

\drawedge[ELpos=50, ELside=r, ELdist=-3, syo=-1, curvedepth=6](q0,q2){\rotatebox{-30}{send, 1}}
\drawedge[ELpos=65, ELside=l, ELdist=.5, syo=-1, curvedepth=6](q0,q2){$\frac{1}{10}$}   

\drawedge[ELpos=50, ELside=r, ELdist=1, curvedepth=-6](q1,q0){ack, 2}
\drawedge[ELpos=50, ELside=l, ELdist=1, curvedepth=6](q2,q0){ack, 5}



}
\end{picture}

%% file: figures/aut-network2.tex
\unitlength=.8mm
\def\fsize{\normalsize}

\begin{picture}(70,50)(0,0)

{\fsize

\node[Nmarks=i](q0)(17,25){$q'_0$}
\node[Nmarks=n](q1)(52,45){$q'_1$}
\node[Nmarks=n](q2)(52,5){$q'_2$}

\drawedge[ELpos=50, ELside=l, ELdist=-3, syo=1, curvedepth=-6](q0,q1){\rotatebox{30}{send, 5}}
\drawedge[ELpos=65, ELside=r, ELdist=.5, syo=1, curvedepth=-6](q0,q1){$\frac{99}{100}$}

\drawedge[ELpos=50, ELside=r, ELdist=-3, syo=-1, curvedepth=6](q0,q2){\rotatebox{-30}{send, 5}}
\drawedge[ELpos=65, ELside=l, ELdist=.5, syo=-1, curvedepth=6](q0,q2){$\frac{1}{100}$}

\drawedge[ELpos=50, ELside=r, ELdist=1, curvedepth=-6](q1,q0){ack, 1}
\drawedge[ELpos=50, ELside=l, ELdist=1, curvedepth=6](q2,q0){ack, 20}



}
\end{picture}

%% file: figures/reducibility.tex
\def\fsize{\normalsize}

\hrule
\begin{picture}(140,50)(0,0)

{\fsize

\put(-1.5,0)
{


\node[Nframe=n,Nadjust=wh,Nmarks=n](asli-zls)(107,47){\asli\/ $\leftrightarrow$ \zls}
\node[Nframe=n,Nadjust=wh,Nmarks=n](asla)(11,35){\asla}
\node[Nframe=n,Nadjust=wh,Nmarks=n](zla)(35,35){\zla}
\node[Nframe=n,Nadjust=wh,Nmarks=n](nla)(59,35){\nla}
\node[Nframe=n,Nadjust=wh,Nmarks=n](zli)(83,35){\zli}
\node[Nframe=n,Nadjust=wh,Nmarks=n](nls)(107,35){\nls}
\node[Nframe=n,Nadjust=wh,Nmarks=n](asls)(131,35){\asls}
\node[Nframe=n,Nadjust=wh,Nmarks=n](dla)(35,25){\dla}
\node[Nframe=n,Nadjust=wh,Nmarks=n](ndli)(83,25){\ndli}
\node[Nframe=n,Nadjust=wh,Nmarks=n](dls)(107,25){\dls}
\node[Nframe=n,Nadjust=wh,Nmarks=n](nbw)(131,25){\nbw}
\node[Nframe=n,Nadjust=wh,Nmarks=n](dbw)(119,15){\dbw}
\node[Nframe=n,Nadjust=wh,Nmarks=n](ndmax)(95,15){\ndmax}
\node[Nframe=n,Nadjust=wh,Nmarks=n](ndcw)(107,5){\ndcw}

\drawedge[ELpos=55, ELside=l, ELdist=1, curvedepth=7, exo=-2](ndcw,ndli){}
\drawedge[ELpos=55, ELside=l, ELdist=1, curvedepth=-6, exo=1](ndcw,nbw){}


\drawedge[ELpos=55, ELside=l, ELdist=1, curvedepth=0](ndmax,ndli){}
\drawedge[ELpos=55, ELside=l, ELdist=1, curvedepth=0](ndmax,dls){}
\drawedge[ELpos=55, ELside=l, ELdist=1, curvedepth=0](dbw,dls){}
\drawedge[ELpos=55, ELside=l, ELdist=1, curvedepth=0](dbw,nbw){}

\drawedge[ELpos=55, ELside=l, ELdist=1, curvedepth=0](zli,asli-zls){}
\drawedge[ELpos=55, ELside=l, ELdist=1, curvedepth=0](nls,asli-zls){}
\drawedge[ELpos=55, ELside=l, ELdist=1, curvedepth=0](asls,asli-zls){}
\drawedge[ELpos=55, ELside=l, ELdist=1, curvedepth=0](dla,asla){}
\drawedge[ELpos=55, ELside=l, ELdist=1, curvedepth=0](dla,zla){}
\drawedge[ELpos=55, ELside=l, ELdist=1, curvedepth=0](dla,nla){}
\drawedge[ELpos=55, ELside=l, ELdist=1, curvedepth=0](ndli,nla){}
\drawedge[ELpos=55, ELside=l, ELdist=1, curvedepth=0](ndli,zli){}
\drawedge[ELpos=55, ELside=l, ELdist=1, curvedepth=0](ndli,nls){}
\drawedge[ELpos=55, ELside=l, ELdist=1, curvedepth=0](dls,nls){}
\drawedge[ELpos=55, ELside=l, ELdist=1, curvedepth=0](nbw,nls){}

\drawedge[ELpos=45, ELside=r, ELdist=1, curvedepth=0, dash={.5 .8}0](nla,zla){?}


\node[Nframe=n,Nadjust=wh,Nmarks=n](ndmax)(95,15){}
\drawedge[ELpos=55, ELside=l, ELdist=1, curvedepth=7, sxo=0, syo=-3](ndmax,dla){}

\drawline[AHnb=0,dash={1.48}0](2,5)(95,5)(125,30)(140,30){}
\node[Nframe=n,Nadjust=wh,Nmarks=n](label)(1.9,8){\makebox(0,0)[l]{quantitative}}
\node[Nframe=n,Nadjust=wh,Nmarks=n](label)(2.2,2.5){\makebox(0,0)[l]{boolean}}




}
}

\end{picture}
\hrule

%% file: figures/aut-ZLa.tex
\unitlength=.8mm
\def\fsize{\normalsize}

\begin{picture}(83,30)(0,0)

{\fsize

\node[Nmarks=i](x0)(12,10){$q_0$}
\node[Nmarks=i, ExtNL=y, NLangle=270, NLdist=2](x0)(12,10){$\wt = 0$}

\node[Nmarks=n](x1)(42,10){$q_1$}
\node[Nmarks=n, ExtNL=y, NLangle=270, NLdist=2](x1)(42,10){$\wt = 1$}

\node[Nmarks=n](x2)(72,10){{\sf sink}}
\node[Nmarks=n, ExtNL=y, NLangle=270, NLdist=2](x2)(72,10){$\wt = 0$}

\drawloop[ELside=l, ELdist=1, loopCW=y, loopdiam=7, loopangle=90](x0){$a,b,\frac{1}{2}$}
\drawedge[ELpos=50, ELside=l, ELdist=1, curvedepth=0](x0,x1){$a,b,\frac{1}{2}$}
\drawloop[ELside=l, ELdist=1, loopCW=y, loopdiam=7, loopangle=90](x1){$b,1$}
\drawedge[ELpos=50, ELside=l, ELdist=1, curvedepth=0](x1,x2){$a,1$}
\drawloop[ELside=l, ELdist=1, loopCW=y, loopdiam=7, loopangle=90](x2){$a,b,1$}



}
\end{picture}

%% file: figures/aut-AsLa.tex
\unitlength=.8mm
\def\fsize{\normalsize}

\begin{picture}(52,30)(0,0)
\put(14,0)
{


{\fsize

\node[Nmarks=i](x0)(12,10){$q_0$}
\node[Nmarks=i, ExtNL=y, NLangle=270, NLdist=2](x0)(12,10){$\wt = 0$}

\node[Nmarks=n](x2)(42,10){{\sf sink}}
\node[Nmarks=n, ExtNL=y, NLangle=270, NLdist=2](x2)(42,10){$\wt = 1$}

\drawloop[ELside=l, ELdist=1, loopCW=y, loopdiam=7, loopangle=130](x0){$b,1$}
\drawloop[ELside=l, ELdist=1, loopCW=y, loopdiam=7, loopangle=50](x0){$a,\frac{1}{2}$}
\drawedge[ELpos=50, ELside=r, ELdist=1, curvedepth=0](x0,x2){$a,\frac{1}{2}$}
\drawloop[ELside=l, ELdist=1, loopCW=y, loopdiam=7, loopangle=90](x2){$a,b,1$}



}
}

\end{picture}

%% file: figures/aut-ZL.tex
\unitlength=.8mm
\def\fsize{\normalsize}

\begin{picture}(88,34)(0,0)

{\fsize

\node[Nmarks=i, iangle=270](x0)(42,10){$q_0$}
\node[ExtNL=y, NLangle=270, NLdist=7](x0)(42,10){$\wt = 1$}

\node[Nmarks=n](x1)(72,10){$q_1$}
\node[Nmarks=n, ExtNL=y, NLangle=270, NLdist=2](x1)(72,10){$\wt = 1$}

\node[Nmarks=n](x2)(12,10){{\sf sink}}
\node[Nmarks=n, ExtNL=y, NLangle=270, NLdist=2](x2)(12,10){$\wt = 0$}

\drawloop[ELside=l, ELdist=1, loopCW=y, loopdiam=7, loopangle=90](x0){$a,\frac{1}{2}$}
\drawedge[ELpos=50, ELside=l, ELdist=1, curvedepth=6](x0,x1){$a,\frac{1}{2}$}
\drawedge[ELpos=50, ELside=l, ELdist=1, curvedepth=0](x0,x2){$b,1$}

\drawloop[ELside=l, ELdist=1, loopCW=y, loopdiam=7, loopangle=90](x1){$a,1$}
\drawedge[ELpos=50, ELside=l, ELdist=1, curvedepth=6](x1,x0){$b,1$}

\drawloop[ELside=l, ELdist=1, loopCW=y, loopdiam=7, loopangle=90](x2){$a,b,1$}



}
\end{picture}